\documentclass{llncs}

\pdfoutput=1
\usepackage{amssymb}
\usepackage{amsmath}
\usepackage[all]{xy}

\newtheorem{obs}[theorem]{Observation}

\newcommand{\mathCommandFont}[1]{\mathrm{#1}}

\newcommand{\complexityClassFont}[1]{\mathrm{#1}}
\newcommand{\problemFont}[1]{\textsc{#1}}

\newcommand{\N}{\ensuremath{\mathbb{N}}}

\newcommand{\mc}[1]{\mathcal{#1}}

\newcommand{\btheorems}[1]{\protect\ensuremath{\mathCommandFont{Th}\kern-2\mu\big(\kern-2\mu#1\kern-2\mu\big)}}
\newcommand{\Btheorems}[1]{\protect\ensuremath{\mathCommandFont{Th}\kern-2\mu\Big(\kern-2\mu#1\kern-2\mu\Big)}}

\newcommand{\class}[1]{\protect\ensuremath\complexityClassFont{#1}}

\newcommand{\co}{\class{co}}

\renewcommand{\L}{\class{L}}
\newcommand{\NL}{\class{NL}}
\renewcommand{\P}{\class{P}}
\newcommand{\NP}{\class{NP}}
\newcommand{\coNP}{\class{coNP}}
\newcommand{\MA}{\class{MA}}
\newcommand{\QCMA}{\class{QCMA}}

\newcommand{\AC}[1]{\class{AC}^{#1}}
\newcommand{\pAC}[1]{\polylog\class{AC}^{#1}}
\newcommand{\NC}[1]{\class{NC}^{#1}}

\newcommand{\SAC}[1]{\class{SAC}^{#1}}

\newcommand{\prob}[1]{\protect\ensuremath\problemFont{#1}}

\newcommand{\TAUT}{\prob{Taut}}
\newcommand{\TTAUT}{\prob{2TAUT}}
\newcommand{\REACH}{\prob{Reach}}
\newcommand{\EXACTOR}{\prob{Exact-Or}}
\newcommand{\EXACTCNT}{\prob{Exact-Count}}

\newcommand{\MAJ}{\prob{Maj}}
\newcommand{\EXMAJ}{\prob{ExMaj}}
\newcommand{\EXMAJEVEN}{\prob{ExMajEven}}
\newcommand{\EQONES}{\prob{EqualOnes}}

\newcommand{\GI}{\prob{GI}}

\newcommand{\REGSCC}{\prob{Reg-SCC}}
\newcommand{\REGNCPS}{\prob{Reg-}\NC{0}\prob{-PS}}
\newcommand{\Th}{\prob{Th}}
\newcommand{\calC}{\mathcal{C}}
\newcommand{\EQUAL}{\prob{Equal}}
\newcommand{\FEAS}[2]{\prob{Feasible}_{#1,#2}}
\newcommand{\CONS}[2]{\prob{Consistent}_{#1,#2}}

\newcommand{\LCONS}[1]{\prob{ConsistentLeaf}_{#1}}

\newcommand{\USTCONN}{\prob{uSTConn}}
\newcommand{\CYCLES}{\prob{Cycles}}
\newcommand{\UNREACH}{\prob{UnReach}}

\newcommand{\poly}{\mathrm{poly}}

\newcommand{\Diag}{\mathrm{Diag}}
\newcommand{\polylog}{\mathrm{poly}\log}
\newcommand{\pACps}{$\polylog\AC{0}$ proof system}
\newcommand{\Up}{\mathrm{UpClose}}
\newcommand{\Minterms}{\mathrm{Minterms}}
\newcommand{\Span}{\mathrm{Span}}

\newenvironment{mylemma}[1]{\rule{0pt}{10pt}\\
  {\bf Lemma~#1.} \it{ }{}}

\newenvironment{myobs}[1]{\rule{0pt}{10pt}\\
  {\bf Observation~#1.} \it{ }{}}

\newenvironment{mytheorem}[1]{\rule{0pt}{10pt}\\
  {\bf Theorem~#1.} \it{ }{}}

\begin{document}

\title{Small Depth Proof Systems}
\author{Andreas Krebs\inst{1} \and
Nutan Limaye\inst{2} \and
Meena Mahajan\inst{3} \and
Karteek Sreenivasaiah\inst{3} 
}

\institute{University of T\"ubingen, Germany
\and
Indian Institute of Technology, Bombay, India
\and
The Institute of Mathematical Sciences, Chennai, India
}

\maketitle

\begin{abstract}
        A proof system for a language $L$ is a function $f$ such that
        Range$(f)$ is exactly $L$.  In this paper, we look at proof
        systems from a circuit complexity point of view and study
        proof systems that are computationally very restricted. The
        restriction we study is: they can be computed by bounded fanin
        circuits of constant depth ($\NC{0}$), or of $O(\log \log n)$
        depth but with $O(1)$ alternations ($\poly\log\AC{0}$). Each
        output bit depends on very few input bits; thus such proof
        systems correspond to a kind of local error-correction on a
        theorem-proof pair.

        We identify exactly how much power we need for proof systems
        to capture all regular languages.  We show that all regular
        language have $\poly\log\AC{0}$ proof systems, and from a
        previous result (Beyersdorff et al, MFCS 2011, where $\NC{0}$
        proof systems were first introduced), this is tight. Our
        technique also shows that $\MAJ$ has $\poly\log\AC{0}$ proof
        system.

        We explore the question of whether $\TAUT$ has $\NC{0}$ proof
        systems. Addressing this question about $\TTAUT$, and since
        $\TTAUT$ is closely related to reachability in graphs, we ask
        the same question about Reachability.  We show that both
        Undirected Reachability and Directed UnReachability have
        $\NC{0}$ proof systems, but Directed Reachability is still
        open.

        In the context of how much power is needed for proof systems
        for languages in $\NP$, we observe that proof systems for a
        good fraction of languages in $\NP$ do not need the full power
        of $\AC{0}$; they have $\SAC{0}$ or $\co\SAC{0}$ proof
        systems.
\end{abstract}

\section{Introduction}
\label{sec:intro}

Let $f$ be any computable function mapping strings to strings. Then
$f$ can be thought of as a proof system for the language
$L=\mbox{range}(f)$ in the following sense: to prove that a word $x$
belongs to $L$, provide a word $y$ that $f$ maps to $x$.  That is, 
view $y$ as a proof of the statement ``$x \in L$'', and
computing $f(y)$ is then tantamount to verifying the proof. From the
perspective of computational complexity, interesting proof systems are
those functions that are efficiently computable and have succinct
proofs for all words in their range. If we use polynomial-time
computable as the notion of efficiency, and polynomial-size as the
notion of succinctness, then $\NP$ is exactly the class of languages
that have efficient proof systems with succinct proofs. For
  instance, the $\coNP$-complete language $\TAUT$ has such proof
  systems if and only if $\NP$ equals $\coNP$
\cite{Cook71}.

Since we do not yet know whether or not $\NP$ equals
$\co$-$\NP$, a reasonable question to ask is how much more
computational power and/or non-succinctness is needed before we can
show that $\TAUT$ has a proof system. For instance, allowing the
verifier the power of randomized polynomial-time computation on
polynomial-sized proofs characterizes the class $\MA$; allowing
quantum power characterizes the class $\QCMA$; one could also allow
the verifier access to some advice, yielding non-uniform classes; see
for instance \cite{Hir10,HI10,CK07,Pudlak09}.

An even more interesting, and equally reasonable, approach is to ask:
how much do we need to reduce the computational power of the verifier
before we can formally establish that $\TAUT$ does not have a proof
system within those bounds?  This approach has seen a rich body of
results, starting from the pathbreaking work of Cook and Reckhow
\cite{CR79}. The common theme in limiting the verifier's power is to
limit the nature of proof verification, equivalently, the syntax of
the proof; for example, proof systems based on resolution, Frege
systems, and so on.  See \cite{BP01,Seg07} for excellent surveys on
the topic.

Instead of restricting the proof syntax, if we only restrict the
computational power of the verifier, it is not immediately obvious
that we get anywhere. This is because it is already known that $\NP$
is characterised by succinct proof systems with extremely weak
verifiers, namely $\AC{0}$ verifiers. Recall that in $\AC{0}$ we
cannot even check if a binary string has an odd number of 1s
\cite{FSS84,Has86}. But an $\AC{0}$ computation can verify that a
given assignment satisfies a Boolean formula. Nonetheless, one can
look for verifiers even weaker than $\AC{0}$; this kind of study was
initiated in \cite{BDKMSSTV13} where $\NC{0}$ proof systems were
investigated. In an $\NC{0}$ proof system, each output bit depends on
just $O(1)$ bits of the input, so to enumerate $L$ as the range of an
$\NC{0}$ function $f$, $f$ must be able to do highly local corrections
to the alleged proof while maintaining the global property that the
output word belongs to $L$. Unlike with locally-decodable
error-correcting codes, the correction here must be deterministic and
always correct. This becomes so restrictive that even some
very simple languages, that are regular and in $\AC{0}$, do not have
such proof systems, even allowing non-uniformity. And yet there is an
$\NP$-complete language that has a uniform $\NC{0}$ proof
system (See \cite{CM01}). (This should not really be that surprising, because it is
known that in $\NC{0}$ we can compute various cryptographic
primitives.)  So the class of
languages with $\NC{0}$ proof systems slices vertically across
complexity classes. It is still not known whether $\TAUT$ has a
(possibly non-uniform) $\NC{0}$ proof
system. Figure~\ref{fig:containment} shows the relationships between
classes of languages with proof systems of the specified kind. (Solid
arrows denote proper inclusion, dotted lines denotes
incomparability.) 
  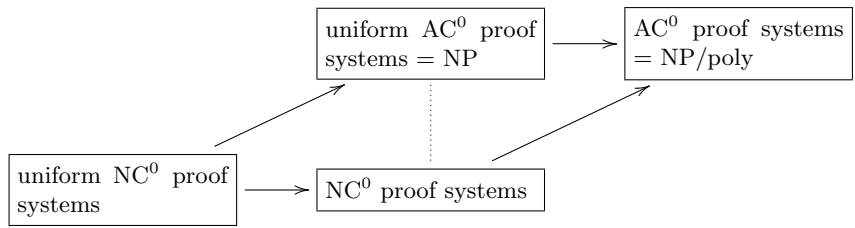
\begin{figure}[ht]
  \label{fig:containment}
  \begin{center}
  \xymatrix{ 
  &
  \fbox{\parbox{28mm}{uniform $\AC{0}$ proof systems = $\NP$}}
  \ar[r] 
  & 
  \fbox{\parbox{28mm}{$\AC{0}$ proof systems = $\NP$/poly}}
  \\
  \fbox{\parbox{28mm}{uniform $\NC{0}$ proof systems}}
  \ar[r] \ar[ru] 
  & 
  \fbox{\parbox{28mm}{$\NC{0}$ proof systems}} \ar@{..}[u]
  \ar[ru] 
  }
  \end{center}
  \caption{Some constant-depth proof systems}
  \end{figure}

The work in \cite{BDKMSSTV13} shows that languages of varying
complexity (complete for $\NC{1}$, $\P$, $\NP$) have uniform $\NC{0}$
proof systems, while the languages $\EXACTOR$, $\MAJ$ amongst others
do not have even non-uniform $\NC{0}$ proof systems. It then focuses
on regular languages, and shows that a large subclass of regular
languages has uniform $\NC{0}$ proof systems. This work takes off from
that point.

\subsection*{Our Results}
We address the question of exactly how much computational power is
required to capture all regular languages via proof systems, and
answer this question exactly. One of our main results (Theorem~\ref{thm:regularPS}) is that
 every regular
language has a proof system computable by a circuit with bounded fanin
gates, depth $O(\log \log n)$, and $O(1)$ alternations. Equivalently,
the proof system is computable by an $\AC{0}$ circuit where each gate
has fanin $(\log n)^{O(1)}$; we refer to the class of such circuits as 
$\pAC{0}$ circuits. By the result of \cite{BDKMSSTV13}, 
$\EXACTOR$ requires depth $\Omega(\log \log n)$, so (upto constant
multiplicative factors) this is tight.  Our proof technique also
generalises to show that $\MAJ$ has $\pAC{0}$ proof systems (Theorem~\ref{thm:thresholdPS}).

The most intriguing question here, posed in \cite{BDKMSSTV13}, is to
characterize the regular languages that have $\NC{0}$ proof
systems. We state a conjecture for this characterization; the
conjecture throws up more questions regarding decidability of some
properties of regular languages.

We believe that $\TAUT$ does not have $\AC{0}$ proof systems because 
otherwise $\NP = \coNP$ (See \cite{Cook71}). As a weaker step, can we
at least
prove that it does not have $\NC{0}$ proof systems? Although it seems that this
should be possible,  we have not yet succeeded. So we ask the same
question about $\TTAUT$, which is in $\NL$, and hence may well have an $\NC{0}$ proof system.
The standard $\NL$ algorithm for $\TTAUT$ is via a
reduction to $\REACH$. So it is interesting to ask -- does $\REACH$ have an
$\NC{0}$ proof system? We do not know yet. 
However in our other main result, we show that undirected $\REACH$, a
language complete for $\L$, has an $\NC{0}$ proof system
(Theorem~\ref{thm:ustconn}). Our construction
relies on a careful decomposition of 
even-degrees-only graphs (established in the proof of Theorem~\ref{thm:cycles}) 
that may be of independent interest. We also show that directed unreachability has an $\NC{0}$
proof system (Theorem~\ref{thm:stunreach}).

Finally, we observe that Graph Isomorphism
does not have $\NC{0}$ proof systems. We also note that for every language $L$
in $\NP$, the language $(\{1\} \cdot L \cdot \{0\}) \cup 0^* \cup 1^*$ has
both $\SAC{0}$ and $\co\SAC{0}$ proof systems (Theorem~\ref{thm:SAC}).

\section{Preliminaries}
\label{sec:prelim}

Unless otherwise stated, we consider only bounded fanin circuits over
$\vee,\wedge,\neg$. 

\begin{definition}[\cite{BDKMSSTV13}]
\label{def:ps}
A circuit family $\{C_n\}_{n>0}$ is a proof system for a language $L$
if there is a function $m:\N \longrightarrow \N$ such that for each
$n$ where $L^{=n} \neq \emptyset$, 
\begin{enumerate}
\item $C_n$ has $m(n)$ inputs and $n$ outputs, 
\item for each $y\in L^{=n}$, there is an $x\in \{0,1\}^{m(n)}$ such
  that $C_n(x)=y$ {\em (completeness)}, 
\item for each $x\in \{0,1\}^{m(n)}$, $C_n(x) \in L^{=n}$ {\em (soundness)}.
\end{enumerate}
\end{definition}

Note that the parameter $n$ for $C_n$ is the number of output bits,
not input bits.
$\NC{0}$ proof systems are proof systems as above where the circuit
has $O(1)$ depth. The definition implies that the circuits are of
linear size.  $\AC{0}$ proof systems are proof systems as above where
the circuit $C_n$ has $O(\log n )$ depth but $O(1)$ alternations
between gate types.
Equivalently, they are proof systems as above of
$n^{O(1)}$ size with unbounded fanin gates and depth $O(1)$.

\begin{proposition}[\cite{BDKMSSTV13}]
\label{prop:old-upper}
A regular language $L$ satisfying any of the following has an $\NC{0}$
proof system:
\begin{enumerate}
\item
$L$ has a  strict star-free expression (built from $\epsilon$, 
$a$, 
and $\Sigma^*$, using concatenation and union).
\item
$L$ is accepted by an automaton with a universally reachable
absorbing final state.
\item
$L$ is accepted by a strongly connected automaton.
\end{enumerate}
\end{proposition}
\begin{proposition}[\cite{BDKMSSTV13}]
\label{prop:old-lower}
\begin{enumerate}
\item
Proof systems for $\MAJ$ need $\omega(1)$ depth. 
\item
Proof systems for $\EXACTCNT^n_k$ and $\neg \Th^n_{k+1}$ need
$\Omega(\log (\log n - \log k))$ depth. In particular, proof systems
for $\EXACTOR$ and for $\EXACTOR\cup 0^*$ need $\Omega(\log \log n)$
depth.
\end{enumerate}
\end{proposition}

\section{Proof Systems for Regular Languages}
\label{sec:regular}

We first explore the extent to which the structure of regular languages can
be used to construct $\NC{0}$ proof systems. At the base level, 
we know that 
all finite languages have $\NC{0}$ proof systems. Building
regular expressions involves unions, concatenation, and Kleene
closure. And the resulting class of regular languages is also closed
under many more operations. We examine these operations one by one.

\begin{theorem}
\label{thm:closures}
\samepage{
Let $\calC$ denote the class of languages  with $\NC{0}$ proof
systems.  Then $\calC$ is closed under 
\begin{multicols}{2}
\begin{enumerate}
\item finite union \cite{BDKMSSTV13},
\item concatenation with finite sets \cite{BDKMSSTV13}, 
\item reversal, 
\item fixed-length morphisms,
\columnbreak
\item inverses of fixed-length morphisms,
\item fixed-length regular transductions computed by strongly
  connected (nondeterministic) finite-state automata.
\end{enumerate}
\end{multicols}
}
\end{theorem}

\begin{proof}
Closure under reversal is trivial. 

Let $h$ be a fixed-length morphism $h:\{0,1\} \longrightarrow
\{0,1\}^k$ for some fixed $k$. Given a proof system $(C_n)$ for $L$, a
proof system $(D_n)$ for $h(L)$ consists of $n$ parallel applications
of $h$ to the each bit of the output of the circuit $C_n$.  Given a
proof system $D'_n$ for $L$, a proof system $C'_n$ for $h^{-1}(L)$
consists of $n$ parallel applications of $h^{-1}$ applied to disjoint
$k$-length blocks of the output of the circuit $D'_{kn}$. $C'_n$ needs
additional input for each block to choose  between possibly
multiple pre-images. 

If $L$ has an $\NC{0}$ proof system $(C_n)$ and $h$ is a regular
transduction computed by a strongly connected automaton
$M$, the construction from \cite{BDKMSSTV13}
(Proposition~\ref{prop:old-upper} (3)) with the output $w$ of $C_n$ as
input will produce a word $x\in L(M)$. A small modification allows us
output the transduction $h(x)$ instead of $x$. This works provided
there are constants $k,\ell$ such that each edge in $M$ is labeled by
a pair $(a,b)$ with $a\in \{0,1\}^k$ and $b\in \{0,1\}^\ell$.
\qed\end{proof}

\begin{theorem}
\label{thm:nonclosures}
Let $\calC$ denote the class of languages  with $\NC{0}$ proof
systems.   $\calC$ is not closed under 
\begin{multicols}{2}
\begin{enumerate}
\item complementation \cite{BDKMSSTV13},
\item concatenation, 
\item symmetric difference, 
\item cyclic shifts, 
\item permutations and shuffles,
\item intersection,
\item quotients.
\end{enumerate}
\end{multicols}
\end{theorem}
\begin{proof}
As noted in \cite{BDKMSSTV13}, $\Th^n_2$ has an $\NC{0}$ proof system
but its complement $\EXACTOR\cup 0^*$ does not. 
The languages denoted by the regular expressions $1$, $0^*$, $10^*$,
and the languages $\Th_1$, $\Th_2$ all have $\NC{0}$ proof systems.
The language $\EXACTOR$ does not, but it can be written as $0^* \cdot
10*$ (concatenation), as $\Th_2 \Delta \Th_1$ (symmetric difference),
as the result of cyclic shifts or permutations on $10^*$, and as the
shuffle of $1$ and $0^*$.

To see the last two non-closures, it is easier to use non-binary
alphabets; the coding back to $\{0,1\}$ is straightforward.  Over the
alphabet $\{0,1,a,b\}$, the languages $\left( 0^*10^* \cup
(0+1+a)^*a(0+1+a)^*\right) $ and $\left( 0^*10^* \cup
(0+1+b)^*b(0+1+b)^*\right) $ both have $\NC{0}$ proof systems (this
follows from Proposition~\ref{prop:old-upper} (2)), but their
intersection is $\EXACTOR$. Also, consider the languages $A=a0^*$,
$B_1=\{ xay \mid |x|=|y|, x \in \EXACTOR, y\in 0^*\}$, $B_2 = \{xay
\mid |x|=|y|, x \in (0+1)^*, y\in \Th_1 \}$. Then $A$ and $B=B_1\cup
B_2$ have $\NC{0}$ proof systems but $\EXACTOR = B \mid A$.\\ 
(A proof system for $B$ is as follows: the input proof at
length $2n+1$ consists of a word $w\in (0+1)^n$ and the sequence of $n$
states $q_1, \ldots , q_n$ allegedly seen by an automaton $M$ for
$\EXACTOR$ on reading $w$. The circuit copies $w$ into $x$. If
$q_{i-1},w_i,q_i$ is consistent with $M$, then it sets $y_i$ to 0,
otherwise it sets $y_i$ to 1. It can be verified that the range of
this circuit is exactly $B^{=2n+1}$.)
\qed\end{proof}

A natural idea is to somehow use the structure of the syntactic monoid
(equivalently, the unique minimal deterministic automaton) to decide
whether or not a regular language has an $\NC{0}$ proof system, and if
so, to build one. Unfortunately, this idea collapses at once: the
languages $\EXACTOR$ and $\Th_2$ have the same syntactic monoid; by
Proposition~\ref{prop:old-lower}, $\EXACTOR$ has no $\NC{0}$ proof
system; and by Proposition~\ref{prop:old-upper} $\Th_2$ has such a
proof system.

The next idea is to use the structure of a well-chosen
(nondeterministic) automaton for the language to build a proof system;
Proposition~\ref{prop:old-upper} does exactly this. It describes two
possible structures that can be used. However, one is subsumed
in the other; see Observation~\ref{obs:absorb-scc} below.

\begin{obs}
\label{obs:absorb-scc}
Let $L$ be accepted by an automaton with a universally reachable
absorbing final state. Then $L$ is accepted by a strongly connected
automaton.
\end{obs}
\begin{proof}
Let $M$ be the non-deterministic automaton with universally reachable
and absorbing final state $q$. That is, $q$ is an accepting state such
that (1)~$q$ is reachable from every other state of $M$, and (2)~there
is a transition from $q$ to $q$ on every letter in $\Sigma$. 
Add $\epsilon$-moves from $q$ to every state of $M$ to get automaton
$M'$. Then $M'$ is strongly connected, and $L(M')=L(M)$. 
\qed\end{proof}


A small generalisation beyond strongly connected automata is automata
with exactly two strongly connected components. However, the automaton
for $\EXACTOR$ is like this, so even with this small extension, we can
no longer construct $\NC{0}$ proof systems. (In fact, we need as much
as $\Omega(\log \log n)$ depth.)

Finite languages do not have strongly connected automata. But they are
strict star-free and hence have $\NC{0}$ proof systems. Strict
star-free expressions lack non-trivial Kleene closure. What can we say
about their Kleene closure? It turns out that for any regular
language, not just a  strict-star-free one, the Kleene closure has an
$\NC{0}$ proof system. 

\begin{theorem}
If $L$ is regular, then $L^*$ has an $\NC{0}$ proof system.
\end{theorem}
\begin{proof}
Let $M$ be an automaton accepting $L$, with no useless states. Adding
$\epsilon$ moves from every final state to the start state $q_0$, and
adding $q_0$ to the set of final states, 
gives an automaton $M'$ for $L^*$. 
Now $M'$ is strongly
connected, so Proposition~\ref{prop:old-upper} gives the $\NC{0}$
proof system.
\qed\end{proof}

Based on the above discussion and  known
(counter-) examples, we conjecture the following
characterization. The structure implies the  proof system, but the
converse seems hard to prove.

\begin{conjecture}
Let $L$ be a regular language. The following are equivalent:
\begin{enumerate}
\item $L$ has an $\NC{0}$ proof system. 
\item For some finite $k$, $L = \bigcup_{i=1}^k u_i \cdot L_i \cdot v_i$, where each $u_i,v_i$ is a
  finite word, and each $L_i$ is a regular language accepted by some
  strongly connected automaton. 
\end{enumerate}
\end{conjecture}

An interesting question arising from this is whether
the following languages are decidable: 
\begin{eqnarray*}
\REGSCC &=& \left\{ M \mid \parbox{0.6\textwidth}{
$M$ is a finite-state automaton; $L(M)$ is
  accepted by some  strongly connected finite automaton}\right\} \\
\REGNCPS &=& \left\{ M \mid \parbox{0.6\textwidth}{$M$ is a finite-state automaton; $L(M)$ has an $\NC{0}$
  proof system}\right\} 
\end{eqnarray*}
(Instead of a finite-state automaton, the input language could be
described in any form that guarantees that it is a regular language. )

We now establish one of our main results. $\NC{0}$ is the restriction of
$\AC{0}$ 
where the fanin of
each gate is bounded by a constant. By putting a fanin bound that is
$\omega(1)$ but $o(n^c)$ for every constant $c$ (``sub-polynomial''),
we obtain intermediate classes.  In particular, restricting the fanin
of each gate to be at most $\polylog n$ gives the class that we call
$\pAC{0}$ lying between $\NC{0}$ and $\AC{0}$.  We show that it is
large enough to have proof systems for all regular languages.
As mentioned earlier, Proposition~\ref{prop:old-lower} implies that
this upper bound is tight.

\begin{theorem}
\label{thm:regularPS}
Every regular language has a $\polylog\AC{0}$ proof system.
\end{theorem}
\begin{proof}
Let $A=(Q,\Sigma, \delta, q_0, F)$ be an automaton for $L$.  We assume
that $\Sigma=\{0,1\}$, larger finite alphabets can be suitably
coded. We unroll the computation of $A$ on inputs of length $n$
to get a layered
branching program $B$ with $n+1$ layers numbered 0 to $n$.  (We can
work directly with the automaton, as discussed in the proof idea, 
but this equivalent formulation is useful in proving the next
theorem as well.)  The initial layer of $B$ has just the start node
$s$ which behaves like $q_0$ in the automaton, while every other layer
of the branching program has as many vertices as $|Q|$.  Since $A$ may
have multiple accepting states, we add an extra layer at the end with
a single sink node $t$, and connect all copies of accepting states at
layer $n$ to $t$ by edges labeled 1.  Note that $B$ has the following
properties:
\begin{itemize}
\item Length $l=n+2$.
\item Every layer except the first and last layer has width (number of vertices
	in that layer) $w=|Q|$.
\item Edges are only between consecutive layers. These edges and
	their labelling are according to $\delta$.
\item All edges from layer $i-1$ to layer $i$ are labelled
	either $x_i$ or $\overline x_i$.
\item A word $a=a_1 \ldots a_n$ is accepted by $A$ if and only if $B$
  has a path from $s$ to $t$ (with $n+1$ edges) with all edge labels
  consistent with $a$.
\end{itemize}
 
Any vertex $u\in B$ can be indexed by a two tuple $(\ell,p)$ where
$\ell$ stands for the layer where $u$ appears and $p$ is the position
where $u$ appears within layer $\ell$.

Consider the interval tree $T$ for $(0,n+1]$ described above. The
  input to the proof system consists of a pair of labels $\langle
  u,v\rangle$ for each node in the interval tree. The labels $u,v$
  point to nodes of $B$. For interval $(i,j]$, the labels are of the
    form $u=(i,p)$, $v=(j,q)$.  Since $i,j$ are determined by the node
    in $T$, the input only specifies the pair $\langle p,q\rangle$
    rather than $\langle u,v\rangle$. That is, it specifies a pair of
    states from $A$, as discussed in the proof idea. At the root node,
    the  labeling is hardwired to be $\langle s,t\rangle$.

Given a word $a=a_1 \ldots a_n$ and a labeling as above of the
interval tree, we define feasibility and consistency as follows:
\begin{enumerate}
\item A leaf node $(k-1,k]$ with $k\in [n]$, labeled $\langle p,q\rangle$,  is
\begin{enumerate}
\item {\bf feasible} if there exists an edge from $(k-1,p)$ to $(k,q)$
  in $B$. (That is, there exists $b\in \Sigma$ such that $q\in
  \delta(p,b)$.)
\item {\bf consistent} if there exists an edge from $(k-1,p)$ to
  $(k,q)$ in $B$ labeled $x_k$ if $a_k=1$, labeled $\overline{x_k}$ if
  $a_k=0$. (That is, $q\in \delta(p,a_k)$.)
\end{enumerate}
(The case $k=n+1$ is simpler: feasible and consistent if $p$ is a
final state of $A$.)
\item An internal node $(i,j]$ labeled $\langle p,q\rangle$  is
\begin{enumerate}
\item {\bf feasible} if there exists a path from $(i,p)$ to $(j,q)$
  in $B$. (That is, there exists a word $b\in \Sigma^{j-i}$ such that $q\in
  \tilde{\delta}(p,b)$.)
\item {\bf consistent} if it is feasible, both its children are
  feasible, and the labels $\langle p',q'\rangle$ and $\langle
  p'',q''\rangle$ of its left and right children respectively satisfy:
  $p=p'$, $q=q''$, $q'=p''$.
\end{enumerate}
\item A node is {\bf fully consistent} if all its ancestors (including
  itself) are consistent. 
\end{enumerate}
Since the label at the root of $T$ is hardwired, the root node is
always feasible. But it may not be consistent. 

For each node $(i,j]$ in the interval tree, and each potential
  labeling $\langle p,q\rangle$ for this node, let $u=(i,p)$ and
  $v=(j,q)$. Define the predicate $R(u,v)$ to be 1 if and only if
  there is a path from $u$ to $v$ in $B$. (ie this potential labeling
  is feasible.) Whenever $R(u,v)=1$, fix a partial assignment
  $w_{u,v}$ that assigns $1$ to all literals that occur as labels
  along an aribtrarily chosen path from $u$ to $v)$. Note that
  $w_{u,v}$ assigns exactly $j-i$ bits, to the variables $x_{i+1},
  \ldots , x_j$. We call $w_{u,v}$  the {\bf feasibility witness} for
  the pair $(u,v)$. 

Let $y$ be the output string of the proof system we construct.
A bit $y_k$ of the output $y$ is computed as follows:
Find the lowest ancestor of the node $(k-1,k]$ that is fully
  consistent. 
\begin{itemize}
\item If the leaf node $(k-1,k]$ is fully consistent, output $a_k$.
\item If there is no such node, then the root node is inconsistent.
  Since it is feasible, the word $w_{s,t}$ is defined. Output the
  $k$th bit of $w_{s,t}$.
\item If such a node is found, and it is not the leaf node itself but
  some $(i,j]$ labeled $\langle p,q\rangle$, let $u=(i,p)$ and
    $v=(j,q)$. The word $w_{u,v}$ is defined and assigns a value to
    $x_k$.  Output this value.
\end{itemize}
It follows from this construction that every word $a \in L$ can be
produced as output: give in the proof the word $a$, and label the
interval tree fully consistent with an $s-t$ path of $B$ consistent
with $a$ (equivalently, an accepting run of $A$ on $a$). 

It also follows that every word $y$ output by this construction 
belongs to $L$. On any proof, moving down from the root of the
interval tree, locate the frontier of lowest fully consistent
nodes. These nodes are feasible and correspond to a partition of the
input positions, and the procedure described above outputs a word
constructed by patching together the feasibility-witnesses for each
part.  

To see that the above construction can be implemented in depth $O(\log \log n)$
with $O(1)$ alternations, observe that each of the conditions - feasibility,
consistency and equality of two labels depend on $O(\log w)$ bits. Hence depth
of $O(\log \log n)$ and $O(1)$ alternations suffices for their implementation.

More formally,
define the following set of predicates: 
\begin{itemize}
\item 
$\EQUAL: [w]^2 \longrightarrow \{0,1\}$ the Equality predicate on
$\log w$ bits. 
\item 
For each $0 \le i < j \le n+1$, $\FEAS{i}{j}: [w]^2 \longrightarrow
\{0,1\}$ is the Feasibility predicate with arguments the labels $(p,q)$
at interval $(i,j]$.
\item 
For each $0 \le i < j+1 \le n+1$, $\CONS{i}{j}: [w]^6 \longrightarrow
\{0,1\}$ is the Consistency predicate at an internal node, with
arguments the labels at interval $(i,j]$ and at its children.
\item 
For each $0 < k \le n+1$, $\LCONS{k}: [w]^2 \times \Sigma
\longrightarrow \{0,1\}$ is the Consistency predicate at leaf
$(k-1,k]$ with arguments the label $\langle p,q\rangle$ 
  and the bit $a_k$ at the leaf.
\end{itemize}
All the predicates depend on $O(\log w)$ bits. So a naive
truth-table implementation suffices to compute them in depth $O(\log w)$ with
$O(1)$ alternations.

For any $0 < k\le n+1$, let the nodes on the path from $(k-1,k]$ to
  the root of the interval tree be the intervals
  $(k-1,k]=(i_0,j_0),(i_1,j_1], \ldots , (i_r,j_r]=(0,n+1]$.  Note: $r
          \in O(\log n)$.

Given a labeling of the tree, the output at position $k$ is given by
the expression below.  (It looks ugly, but it is just implementing the
scheme described above.  We write it in this detail to make the
$\pAC{0}$ computation explicit.)
\begin{eqnarray*}
y_k &=& \left[ a_k \wedge \LCONS{k} \wedge \bigwedge_{h=1}^{r} \CONS{i_h}{j_h}\right]\\
&\vee& 
\left[~(w_{s,t})_k\wedge \overline{\CONS{0}{n+1}}~\right]
\\
&\vee& 
\left[
\bigvee_{h=1}^{r}
(w_{(i_h,p_h),(j_h,q_h)})_k
\wedge 
\overline{\CONS{i_{h-1}}{j_{h-1}}} \wedge \bigwedge_{g=h}^{r}
\CONS{i_g}{j_g}
\right]
\end{eqnarray*}
where the arguments to the predicates are taken from the tree
labeling. This computation adds $O(1)$ alternations and $O(\log \log
n)$ depth to the computation of the predicates, so it is in $\pAC{0}$. 
\qed\end{proof}

While proving Theorem~\ref{thm:regularPS}, we unrolled the computation of a
$w$-state automaton on inputs of length $n$ into a layered branching
program BP of width $w$ with $\ell=n+2$ layers. The BP so obtained is
nondeterministic whenever the automaton is. The BP has a very
restricted structure which we exploited to construct the
$\pAC{0}$ proof system.

We observe that some restrictions on the BP structure can be relaxed
and still we can construct a $\pAC{0}$ proof system.  
\begin{definition}
\label{def:restrBP}
A branching program for length-$n$ inputs is {\bf structured} if it
satisfies the following:
\begin{enumerate}
\item It is layered: vertices are partitioned into $n+1$ layers $V_0, \ldots
  , V_n$  and all edges are between adjacent layers $E \subseteq
  \cup_i (V_{i-1} \times V_{i})$.
\item Each layer has the same size $w = |V_i|$, the width of the
  BP. (This is not critical; we can let $w=\max|V_i|$.)
\item There is a permutation $\sigma\in S_n$ such that for $i\in [n]$,
  all edges in $V_{i-1} \times V_{i}$ read  $x_{\sigma(i)}$ or
  $\overline{x_{\sigma(i)}}$.
\end{enumerate}
\end{definition}
Non-uniform automata \cite{Bar89,BT88} give rise to branching programs
that are structured with $w$ the number of states in the
automaton. For instance, the language $\{xx \mid x \in \{0,1\}^*\}$ is
not regular. But if the input bits are provided in the order $1,m+1,
2, m+2, \ldots , m, 2m$ then it can be decided by a finite-state
automaton.  This gives rise to a structured BP where $\sigma$ is the
inverse of the above order.  (eg $r_2=m+1$, $r_3=2$, $\sigma(m+1)=2$,
$\sigma(2)=3$. )

The idea behind the construction in Theorem~\ref{thm:regularPS} works
for such structured BPs. It yields a proof system with depth
$O(\log\log n + \log w)$. This means that for $w\in O(\poly \log n)$,
we still get $\pAC{0}$ proof systems. Potentially, this is much bigger
than the class of languages accepted by non-uniform
finite-state automata. Formally,
\begin{theorem}
\label{thm:BPs}
Languages accepted by structured branching programs of width $w\in
(\log n)^{O(1)}$ have $\pAC{0}$ proof systems. 
\end{theorem}

For the language $\MAJ$ of strings with more 1s than 0s, and in
general for threshold languages $\Th^n_k$ of strings with at least $k$
1s, we know that there are constant-width branching programs, 
but these are not structured in the sense above. It can be shown that a
structured BP for $\MAJ$ must have width $\Omega(n)$ (a family of
growing automata $M_n$ for $\MAJ$, where $M_n$ is guaranteed to be
correct only on $\{0,1\}^n$, must have $1 + n/2$ states in $M_n$).
This is much much more than the $\polylog$ width bound used in the
construction in Theorem~\ref{thm:regularPS}.  Nevertheless, we show
below how we can modify that construction to get a $\pAC{0}$ proof
system even for threshold languages.

\begin{theorem}
For every $n$ and $t \le n$, the language $\Th^n_t$  has a \pACps.
\label{thm:thresholdPS}
\end{theorem}
\begin{proof}
We follow the approach in Theorem~\ref{thm:regularPS}: the input to
the proof system is a word $a=a_1, \ldots , a_n$ and auxiliary
information in the interval tree allowing us to correct the word if
necessary.  The labeling of the tree is different for this language,
and is as follows. Each interval $(i,j]$ in the tree gets a label
  which is an integer in the range $\{ 0,1, \ldots, j-i\}$. The
  intention is that for an input $a=a_1, \ldots , a_n$, this label is
  the number of 1s in the subword $a_{i+1}\ldots a_j$. 
  For thresholds,  we relax the constraint: we expect the label of
  interval $(i,j]$ to be {\bf no more than} the number of 1s in the
    subword.  At a leaf node $(k-1,k]$, we do not give explicit
      labels; $a_k$ serves as the label. At the root also, we do not
      give an explicit label; the label $t$ is hard-wired.  (We
      restrict the label of any interval $(i,j]$ to the range
      $[0,j-i]$, and interpret larger numbers as $j-i$.)

For any node $u$ of $T$, let $l(u)$ denote the label of $u$. A node $u$ with
children $v,w$ is {\bf consistent} if $l(u)\le l(v)+l(w)$. 

Let the output of our proof system be $y_1,\dots,y_n$. The
construction is as follows:
\begin{itemize}
\item If all nodes on the path from $(k-1,k]$ to the root in $T$ are
  consistent, then $y_k=a_k$.
\item Otherwise, $y_k=1$.
\end{itemize}
In analogy with Theorem~\ref{thm:regularPS}, we use here for each
interval $(i,j]$ the feasibility witness $1^{j-i}$, independent of the
  actual labels.
Thus the construction forces this property: at a node $u$
corresponding to interval $(i,j]$  labelled $\ell(u)$, 
the subword $y_{i+1},\dots,y_{j}$ has at least $\min\{\ell(u),j-i\}$ 1s. Thus, the
output word is always in $\Th_t^n$. Every word in $\Th_t^n$ is
produced by the system at least once, on the proof that gives, for
each interval other than $(0,n]$, the number of 1s in the
  corresponding subword.

As before, the $\CONS{i}{j}$ predicate at a node depends on 3 labels,
each of which is $O(\log n)$ bits long. A truth-table implementation
is not good enough; it will give an $\AC{0}$ circuit. But the actual
consistency check only involves adding and comparing $m=\log n$ bit
numbers. Since addition and comparison are in $\AC{0}$, this can be
done in depth $O(\log m)$ with $O(1)$ alternations. Thus the overall
depth is $O(\log \log n)$.
\qed\end{proof}
\begin{corollary}
	For every $n$ and $t\le n$, $\EXACTCNT^n_t$ has a
	$\polylog\AC{0}$ proof system.
\end{corollary}
\begin{proof}
	We follow the same approach as Theorem~\ref{thm:thresholdPS}. We redefine {\bf
	consistent} as follows: 
	For any node $u$ of $T$, let $l(u)$ denote the label of $u$. A node $u$ with
children $v,w$ is consistent if $l(u)= l(v)+l(w)$. 
Let the output of our proof system be $y_1,\dots,y_n$. The
construction is as follows:
\begin{itemize}
\item If all nodes on the path from $(k-1,k]$ to the root in $T$ are
  consistent, then $y_k=a_k$.
\item Otherwise, let $u=(p,q]$ be the topmost node along the path from $(k-1,k]$ to
	the root that is not consistent. We output $y_k=1$ if $k-p\le l(u)$, $0$
	otherwise.
\end{itemize}
That is, for $u=(i,j]$ labeled $\ell(u)$, if $L=\min\{\ell(u),j-i\}$,
use feasibility witness  $1^{L}0^{j-i-L}$.
\qed\end{proof}

\section{$\TTAUT$, Reachability and $\NC{0}$ proof systems}
\label{sec:reach}
In this section, we first look at the language Undirected Reachability,
which is known to be in (and complete for) $\L$
(\cite{Rei2008}). Intuitively, the property of  connectivity is a global
one. However, viewing it from a different angle gives us a way to
construct an $\NC{0}$ proof system for it under
the standard adjacency matrix encoding (i.e., our proof system will output 
adjacency matrices of all graphs that have a path between $s$ and $t$,
and of no other graphs). In the process, we give an NC$^0$ proof system for the set of all
undirected graphs that are a union of edge-disjoint cycles.

Define the following languages:
\[
\USTCONN=\left\{ A \in \{0,1\}^{n \times n} |\ 
\parbox{0.6\textwidth}{
$A$ is the adjacency matrix of an undirected graph $G$ where vertices
  $s=1$, $t=n$ are in the same connected component.
}
\right\}
\]
\[
\CYCLES=\left\{ A \in \{0,1\}^{n \times n} |\ 
\parbox{0.6\textwidth}{
$A$ is the adjacency matrix of an undirected graph $G = (V,E)$ where 
$E$ is the union of edge-disjoint simple cycles.
}
\right\}
\]
(For simplicity, we will say $G \in \USTCONN$ or $G\in \CYCLES$ instead
of referring to the adjacency matrices. )

\begin{theorem}
\label{thm:ustconn}
The language $\USTCONN$ has an NC$^0$ proof system.
\end{theorem}
\begin{proof}
We will need an addition operation on graphs: $G_1\oplus G_2$ denotes
the graph obtained by adding the corresponding adjacency matrices
modulo 2. We also need a notion of upward closure: For any language
$A$, $\Up(A)$ is the language $B = \{ y : \exists x \in A, |x|=|y|,
\forall i, x_i = 1 \implies y_i=1\}$.  In particular, if $A$ is a
collection of graphs, then $B$ is the collection of super-graphs
obtained by adding edges. Note that (undirected) reachability is
monotone and hence $\Up(\USTCONN) = \USTCONN$.

	Let $L_1=\{G=G_1 \oplus (s,t)| G_1\in \CYCLES \}$ and $L_2 =
        \Up(L_1)$. We show:
	\begin{enumerate}
		\item $L_2=\USTCONN$.
		\item If $L_1$ has an NC$^0$ proof system, then $L_2$ has an NC$^0$ proof system.
		\item If $\CYCLES$ has an NC$^0$ proof system, then $L_1$ has an NC$^0$ proof system.
		\item $\CYCLES$ has an NC$^0$ proof system.
	\end{enumerate}
	\textbf{Proof of 1}:  We show that $L_1 \subseteq
        \USTCONN \subseteq L_2$.  Then applying upward closure,
        $L_2 = \Up(L_1) \subseteq
        \Up(\USTCONN)=\USTCONN \subseteq \Up(L_2)=L_2$.  
        
        $L_1\subseteq \USTCONN$: Any graph $G\in L_1$ looks like $G=H\oplus
	(s,t)$, where $H\in$ $\CYCLES$. If $(s,t)\notin H$, then $(s,t)\in G$ and
	we are done. If $(s,t)\in H$, then $s$ and $t$ lie on a cycle
        $C$ and	hence removing the $(s,t)$ edge will still leave $s$ and $t$
	connected by a path $C\setminus \{(s,t)\}$. 

	$\USTCONN \subseteq L_2$:  Let $G\in \USTCONN$.
	Let $\rho$ be an $s$-$t$ path in $G$. Let $H=(V,E)$ be a graph such that
	$E=$ edges in $\rho$. Then, $G\in \Up(\{H\})$. We can write
        $H$ as $H'\oplus (s,t)$ where $H' = H \oplus (s,t) = \rho \cup
        (s,t)$; hence $H' \in \CYCLES$. Hence $H\in L_1$, and so $G\in L_2$.

\noindent \textbf{Proof of 2}: 
We show a more general construction for  monotone properties, and then
use it for $\USTCONN$.  

Recall that  a  function $f$ is monotone if whenever 
$f(x)=1$ and  $y$ dominates $x$ (that is,  $\forall i\in
[n]$, $x_i=1 \Rightarrow y_i=1$), then it also holds that
$f(y)=1$. For such a  function, a string $x$ is a {\em
  minterm} if $f(x)=1$ but $x$ does not dominate any $z$ with
$f(z)=1$. $\Minterms(f)$ denotes the set of all minterms of
$f$. Clearly, $\Minterms(f) \subseteq f^{-1}(1)$.
The following lemma states that for any monotone function $f$, constructing a
proof system for a language that 
sits in between $\Minterms(f)$ and $f^{-1}(1)$ suffices to
get a proof system for $f^{-1}(1)$.
\begin{lemma}
	Let $f:\{0,1\}^*\longrightarrow \{0,1\}$ be a monotone boolean function and
	let $L=f^{-1}(1)$. Let $L_n=L\cap \{0,1\}^n$. 
	Let $L'$ be a language such that for each length $n$,
	$(\mbox{Minterms}(L)\cap\{0,1\}^n)\subseteq (L'\cap \{0,1\}^n) \subseteq
	L_n$. If $L'$ has a proof system of
	depth $d$, size $s$ and $a$ alternations, then $L$ has a proof system
	of depth $d+1$, size $s+n$ and at most $a+1$ alternations.
	\label{lem:upwardclosure}
\end{lemma}

\begin{proof}
	Let $C$ be a proof circuit for $L'$ that takes input string $x$. 
	We construct a proof system for $L$
	using $C$ and asking another input string $y\in\{0,1\}^n$. The $i$'th
	output bit of our proof system is $C(x)_i\vee y_i$.
\qed\end{proof}

Now note that $\Minterms(\USTCONN)$ is exactly the set of graphs where
the edge set is a simple $s$-$t$ path. We have seen that $L_1
\subseteq \USTCONN$. As above, we can see that $H \in
\Minterms(\USTCONN) \implies H \oplus (s,t) \in \CYCLES \implies H \in
L_1$. Statement 2 now follows from Lemma~\ref{lem:upwardclosure}.

\noindent \textbf{Proof of 3}: Let $A$ be the adjacency matrix output by the
	the NC$^0$ proof system for $\CYCLES$. The proof system for $L_1$ outputs
	$A'$ such that $A'[s,t]=\overline{A[s,t]}$ and rest of $A'$ is same as
	$A$.

\noindent \textbf{Proof of 4}: This is of independent interest, and is
proved in theorem~\ref{thm:cycles} below.

This completes
the proof of theorem~\ref{thm:ustconn}. \qed
\end{proof}

We now construct $\NC{0}$ proof systems for the language $\CYCLES$. 
\begin{theorem}
\label{thm:cycles}
	The language $\CYCLES$ 
	has an NC$^0$ proof system.
\end{theorem}
\begin{proof}
	To design an NC$^0$ proof system for $\CYCLES$, we derive our intuition
	from algebra.

Let $\mc{T}$ be a family of graphs. We say that an edge $e$ is \emph{generated}
by a sub-family $\mc{S} \subseteq \mc{T}$ if the number of graphs in $\mc{S}$
which contain $e$ is odd. We say that the family $\mc{T}$ generates a graph $G$ if
there is some sub-family $\mc{S} \subseteq \mc{T}$ such that every edge in $G$ is
generated by $\mc{S}$, and no other edge is generated. We first observe that to generate every graph in the set
$\CYCLES$, we can set $\mc{T}$ to be the set of \emph{all} triangles. Given any
cycle, it is easy to come up with a set of traingles that generates the cycle;
namely, take any triangulation of the cycle. Therefore, if we let $\mc{T}$ be the
set of all triangles on $n$ vertices, it will generate every graph in
$\CYCLES$. Also, no other graph will be generated because any set $\mc{S}
\subseteq  \CYCLES$ generates a set contained in $\CYCLES$ (see Lemma~\ref{lem:cycles-closure} below). This immediately
gives a proof system for $\CYCLES$: given a vector $x \in \binom{n}{3}$, we will
interpret it as a subset $\mc{S}$ of triangles. We will output an edge $e$ if
it is a part of odd number of triangles in $\mc{S}$. Finally, because of the
properties observed above, any graph generated in this way will be a graph from
the set $\CYCLES$.

Unfortunately, this is not an $\NC{0}$ proof system because to decide
if an edge is generated, we need to look at 
$\Omega(n)$ triangles. For designing an $\NC{0}$ proof
system we need to come up with a set of triangles 
such that for any graph $G
\in \CYCLES$, every edge in $G$ is a part of $O(1)$ triangles.

So on the one hand, we want the set of triangles to generate every graph in
$\CYCLES$, and on the other hand we need that for any graph $G \in \CYCLES$, every
edge in $G$ is a part of $O(1)$ triangles. We show that such a set of triangles
indeed exists.

Thus our task now is to find a set of triangles $T\subseteq \CYCLES$ such that:
\begin{enumerate}
	\item Every graph in $\CYCLES$ can be generated using triangles from $T$.
		i.e., \[ \CYCLES \subseteq \Span(T) \triangleq \left\{\sum_{i=1}^{|T|}a_i t_i\ |\ \forall i,
		a_i\in \{0,1\}, t_i\in T\right\}\]
	\item Every graph generated from triangles in  $T$ is in $\CYCLES$;
          $\Span(T) \subseteq \CYCLES$.
	\item $\forall u,v\in [n]$, the edge $(u,v)$ is contained in at most
		6 triangles in $T$.
\end{enumerate}
	Once we find such a set $T$, then our proof system asks as input the
	coefficients $a_i$ which indicate the linear combination needed to
	generate a graph in $\CYCLES$. An edge $e$ is present in the output if, among
	the triangles that contain $e$, an odd number of them have
	coefficient set to $1$ in the input. By property 3, each output edge
	needs to see only constant many input bits and hence the circuit we
	build is NC$^0$. We will now find and describe $T$ in detail.

	Let the vertices of the graph be numbered from $1$ to $n$. Define the length of
	an edge $(i,j)$ as $|i-j|$. A triple
	$\langle i,j,k \rangle$ denotes the set of triangles on 
        vertices $(u,v,w)$ where $|u-v|=i$,
	$|v-w|=j$, and  $|u-w|=k$. We now define the set
	\[ T=\bigcup_{i=1}^{n/2}\langle i,i,2i \rangle \cup \langle i,i+1,2i+1\rangle \]
	\textbf{Observation} It can be seen that $|T|\le \frac{3}{2}n^2$. This
        is linear in the length of the output, which has ${n \choose
        2}$ independent bits.

	We now show that $T$ satisfies all properties listed earlier.\\
	\textbf{$T$ satisfies property 3}: Take any edge $e=(u,v)$. Let its
	length be $l=|u-v|$. $e$ can either be the longest edge in a triangle or one of the
	two shorter ones. If $l$ is even, then $e$ can be the longest edge
	for only $1$ triangle in $T$ and can be a shorter edge in at most $4$ triangles
	in $T$.	If $l$ is odd, then $e$ can be the longest edge for at most $2$
	triangles in $T$ and can be a shorter
	edge in at most $4$ triangles. Hence, any edge is contained in at most
	$6$ triangles.
\noindent \textbf{$T$ satisfies property 2}: 
To see this, note first that $T \subseteq \CYCLES$.  Next, observe the
following closure property of cycles: 
\begin{lemma} For any $G_1,G_2\in$ $\CYCLES$, the graph $G_1\oplus
	G_2\in $ $\CYCLES$.
	\label{lem:cycles-closure}
\end{lemma}
\begin{proof}
        A well-known fact about connected graphs is that they are
        Eulerian if and only if every vertex has even degree. The
        analogue for general (not necessarily connected) graphs is
        Veblen's theorem \cite{Veb1912}, which states that
        $G\in\CYCLES$ if and only if every vertex in $G$ has even
        degree.

        Using this, we see that if for $i \in [2]$, $G_i \in \CYCLES$
        and if we add the adjacency matrices modulo 2, then degrees of
        vertices remain even and so the resulting graph is also in $\CYCLES$. 
\qed\end{proof}
It follows that $\Span(T)\subseteq \CYCLES$. 

\noindent \textbf{$T$ satisfies property 1}: 
		We will show that any graph $G\in\CYCLES$ can be written 
		as a linear combination of triangles in $T$.
	Define, for a graph $G$, the parameter $d(G)=(l,m)$ where $l$ is the
	length of the longest edge in $G$ and $m$ is the number of edges in $G$
	that have length $l$. For graphs $G_1,G_2\in \CYCLES$, with
	$d(G_1)=(l_1,m_1)$ and $d(G_2)=(l_2,m_2)$, we say $d(G_1)<d(G_2)$ if and
	only if either $l_1<l_2$ holds or $l_1=l_2$ and $m_1<m_2$. Note that for
	any graph $G\in\CYCLES$ with $d(G)=(l,m)$, $l\ge 2$.

	\begin{claim}
		Let $G\in \CYCLES$. If $d(G)=(2,1)$, then $G\in T$.
		\label{claim:decompbasecase}
	\end{claim}
	\begin{proof}
		It is easy to see that $G$ has to be a triangle with edge
		lengths $1,1$ and $2$. All such triangles are contained in $T$
		by definition.
	\qed\end{proof}

	\begin{lemma}
		For every $G \in\CYCLES$ with $d(G)=(l,m)$, either $G\in T$ or 
		there is a $t \in T$, and $H \in\CYCLES$ such that $G=H\oplus t$
		and $d(H)<d(G)$.
		\label{lem:decompstep}
	\end{lemma}
	\begin{proof}
		If $G\in T$, then we are done. So now consider the
                case when 
		$G\notin T$:

	Let $e$ be a longest edge in $G$. Let $C$ be the cycle which contains
	$e$. Pick $t\in T$ such that $e$ is the longest edge in $t$.
		$G$ can be written as $H\oplus t$ where $H=G\oplus t$.
		From Lemma~\ref{lem:cycles-closure} and since $T\subseteq
                \CYCLES$, we know that $H\in \CYCLES$. Let $t$ have the edges
		$e,e_1,e_2$. Any edge present in both $G$
		and $t$ will not be present in $H$. Since $e\in G\cap t$,
		$e\notin H$. Length of $e_1$ and $e_2$ are both less than $l$
		since $e$ was the longest edge in $t$. Hence the number of times
		an edge of length $l$ appears in $H$ is reduced by $1$ and the
		new edges added(if any) to $H$ (namely $e_1$ and $e_2$) have
		length less then $l$. Hence if $m>1$, then $d(H)=(l,m-1)< d(G)$.
		If $m=1$, then $d(H)=(l',m')$ for some $m'$ and $l'<l$, and hence
		$d(H)<d(G)$.
	\qed\end{proof}
	
	By repeatedly applying Lemma~\ref{lem:decompstep}, we can obtain the exact
	combination of triangles from $T$ that can be used to give any $G\in 
	\CYCLES$. A more formal proof will proceed by induction on the parameter
	$d(G)$ and each application of Lemma \ref{lem:decompstep} gives a graph $H$
	with a $d(H)<d(G)$ and hence allows for the induction hypothesis to be
	applied. The base case of the induction is given by Lemma
	\ref{claim:decompbasecase}. Hence $T$ satisifes property 1. 

Since $T$ satisfies all three properties, we obtain an $\NC{0}$ proof
system for $\CYCLES$, proving the theorem.
\qed\end{proof}

The above proof does not work for directed $\REACH$.
However, we can show that directed un-reachability can be captured by $\NC{0}$ proof
systems.
\begin{theorem}
	\label{thm:stunreach}
The language $\UNREACH$  defined below has an $\NC{0}$ proof
	system under the standard adjacency matrix encoding.
\[
\UNREACH=\left\{ A \in \{0,1\}^{n \times n} |\ 
\parbox{0.6\textwidth}{
$A$ is the adjacency matrix of a directed graph $G$ with no path from 
  $s=1$ to $t=n$.
}
\right\}
\]
\end{theorem}
\begin{proof}
	As proof, we take as input an adjacency matrix $A$ and an $n$-bit vector
	$X$ with $X(s)=1$ and $X(t)=0$ hardwired. Intuitively, $X$ is like a
	characteristic vector that represents all vertices that can be reached
	by $s$. 

	The  adjacency matrix $B$ output by our proof system is: 
\[
        B[i,j] = \left\{ 
        \textrm{
       \begin{tabular}{rl}
          1 & if 	$A[i,j]=1$ and it is not the case that
          $X(i)=1$ and $X(j)=0$, \\
          0 & otherwise
        \end{tabular}
       }
       \right.
\]

\noindent
	Soundness: No matter what $A$ is, $X$ describes an $s,t$ cut since $X(s)=1$
	and $X(t)=0$. So any gaph output by the proof system will not have a
	path from $s$ to $t$.

\noindent	Completeness: For any $G\in \UNREACH$, use the adjacency matrix
	of $G$ as $A$ and give input $X$ such that $X(v)=1$ for a vertex $v$
        if and only if $v$ is reachable from $s$.
\qed\end{proof}

\section{Pushing the Bounds}
\label{sec:push}

We know that any language in $\NP$ has $\AC{0}$ proof
systems. Srikanth Srinivasan recently showed that $\AC{0}$, or more
precisely $\Omega(\log n)$ depth in a bounded fanin model, is
necessary for some languages in $\NP$.
We sketch his proof below. 
\begin{theorem}[Srikanth Srinivasan (private communication)]
\label{thm:ac0-necessary}
There is a  language $A$ in $\NP$ such that any bounded-fanin proof
system for $A$ needs $\Omega(\log n)$ depth. 
\end{theorem}
\begin{proof}
	Let $A\subseteq \{0,1\}^n$ be an error correcting code of constant rate
	and linear distance that can be efficiently computed. Such codes are
	known to exist. See for example \cite{Jus72}.
	Suppose there is a proof system $C_n:\{0,1\}^m\longrightarrow \{0,1\}^n$
	of depth $d$ that outputs exactly the strings in $A$. 
	Assume that $C$ is non-degenerate. i.e., for every 
	input position $i$, $\exists x\in\{0,1\}^m$ such that $C(x)\neq C(x\oplus e_i)$.
	Note that $m\geq n$ since $A$ is constant rate
	($|A\cap\{0,1\}^n|=2^{\Omega(n)}$). Note that each output bit is a
	function of at most $2^d$ input bits. By an
	averaging argument, there exists an input position $i$ such that
	$x_i$ is connected to at most $2^d$ output positions. For this $i$, let
	$x$ be an input such that $C(x)\neq C(x\oplus e_i)$. But since $C(x)$
	and $C(x\oplus e_i)$ are both codewords in $A$, they must differ in at
	least $2^{\Omega(n)}$ positions since $A$ is has linear distance. This
	implies that $x_i$ is connected to at least $2^{\Omega(n)}$ output
	positions and so $d=\Omega(\log n)$.
\qed\end{proof}

However, we note that proof systems for a big fragment of $\NP$ do not require the full power of
$\AC{0}$. In particular, for every language in $\NP$, an extremely
simple padding yields another language with simpler proof systems.
\begin{theorem}
\label{thm:SAC}
Let $L$ be any language in $\NP$.
\begin{enumerate}
\item
If $L$ contains $0^*$, then $L$ has a proof system where negations
appear only at leaf level, $\wedge$ gates have unbounded fanin, $\vee$
gates have $O(1)$ fanin, and the depth is $O(1)$.  That is, $L$ has a
$\co\SAC{0}$ proof system.
\item
If $L$ contains $1^*$, then $L$ has a proof system where negations
appear only at leaf level, $\vee$ gates have unbounded fanin, $\wedge$
gates have $O(1)$ fanin, and the depth is $O(1)$.  That is, $L$ has an
$\SAC{0}$ proof system.
\item The language $(\{1\} \cdot L \cdot \{0\}) \cup 0^* \cup 1^*$ has
  both $\SAC{0}$ and $\co\SAC{0}$ proof systems.
\end{enumerate}
\end{theorem}
\begin{proof}
Let $L$ be a language in $\NP$. Then there is a family of uniform
polynomial-sized circuits $(C_n)$, where each $C_n$ has $q(n)$ gates,
$n$ standard inputs $x$ and $p(n)$ auxiliary inputs $y$, such that for
each $x\in\{0,1\}^n$, $x\in L \Longleftrightarrow \exists y:
C_n(x,y)=1$.  We use this circuit to construct the proof system. The
input to the proof system consists of words $x=x_1 \ldots x_n$, $y=y_1
\ldots y_{p(n)}$, $z=z_1 \ldots z_{q(n)}$. The intention is that $y$
represents the witness such that $C_n(x,y)=1$, and $z$ represents the
vector of values computed at each gate of $C_n$ on input $x,y$. There
are two ways of doing self-correction with this information:
\begin{itemize}
\item Check for consistency: Check that every gate $g_i = g_j \circ
  g_k$ satisfies $z_i = z_j \circ z_k$.  Output the string $w$ where
  $\langle w\rangle = \langle x\rangle \wedge (\bigwedge_{i=1}^{q(n)}
  [z_i = z_j \circ z_k])$. If even one gate is inconsistent, $w$
  equals $0^*$, otherwise $w$ is the input $x$ that has been certified
  by $y,z$; hence $w$ is in $L \cup 0^*$.  Every string in $L$ can be
  produced by giving    witness $y$ and consistent $z$. The expression shows that this
  is a $\co\SAC{0}$ circuit.
\item Look for an inconsistency: Find a gate $g_i = g_j \circ g_k$
  where $z_i \neq z_j \circ z_k$.  Output the string $w$ where
  $\langle w\rangle = \langle x\rangle \vee (\bigvee_{i=1}^{q(n)} [z_i
    \neq z_j \circ z_k])$. If even one gate is inconsistent, $w$
  equals $1^*$, otherwise $w$ equals the input $x$ that has been
  certified by $y,z$; hence $w$ is in $L \cup 1^*$.  Every string in
  $L$ can be produced by giving suitable $y,z$. The expression shows
  that this is an $\SAC{0}$ circuit.
\end{itemize}
\qed\end{proof}

Ideally, we would like to have a notion of a reduction $\le$ such that
if $A \le B$ and if $A$ needs $\Omega(d)$ depth in proof systems, then
so does $B$. Such a notion was implicitly used in proving
Theorem~\ref{thm:ustconn}; we showed that a lower bound for $\CYCLES$
translated to a lower bound for $\USTCONN$.  However, part 3 of
theorem~\ref{thm:SAC} suggests that for $\NC{0}$ proof systems in
general, such ``reductions'' are necessarily rather fragile, and we do
not yet see what is a reasonable and robust definition to adopt. Using
some reduction-like techniques, we can give depth lower bounds for
proof systems for some more languages. We collect some such results in
Lemma~\ref{lem:GI} below; all start from the hardness of $\MAJ$.

Using Lemma~\ref{lem:upwardclosure} and the known lower bound for
$\MAJ$ from \cite{BDKMSSTV13}, we can show that 
the following languages have no $\NC{0}$ proof systems: 
\begin{lemma}
\label{lem:GI}
The following languages do not have $\NC{0}$ proof systems.
\begin{enumerate}
\item $\EXMAJ$, consisting of strings $x$ with exactly $\lceil |x|/2 \rceil$
1s. 
\item $\EQONES=\{ xy \mid x,y\in \{0,1\}^*, |x|=|y|, |x|_1 = |y|_1 \}$.
\item $\GI = \{G_1,G_2 \mid \text{Graph $G_1$ is isomorphic to graph
  $G_2$}\}$. \\
Here we assume that $G_1$ and $G_2$ are specified via
  their 0-1 adjacency matrices, and that 1s on the diagonal are
  allowed (the graphs may have
  self-loops).
\end{enumerate}
\end{lemma}
\begin{proof}

\begin{enumerate}
\item To show that $\EXMAJ$ does not have  $\NC{0}$
	proof systems, note that:
	\begin{itemize}
		\item The language $\MAJ$ does not have $\NC{0}$ proof systems
			(See \cite{BDKMSSTV13}). 
		\item $\Minterms(\MAJ) = \EXMAJ$; $\MAJ = \Up(\EXMAJ)$. 
		\item Lemma \ref{lem:upwardclosure} now implies
                  $\EXMAJ$ does not have an $\NC{0}$ proof system. 
	\end{itemize}
By the same argument, $\EXMAJ$ restricted to even-length strings, call
it $\EXMAJEVEN$, has
no $\NC{0}$ proof systems. 
\item 	We will show that if $\EQONES$ has an $\NC{0}$ proof system, then so does the
	language $\EXMAJEVEN$.
Consider the slice 
\[\EQONES^{=2n} = \{ xy \mid |x|=|y| = n;
  \text{~~$x$ and $y$ have an equal number of 1s}\}.\] 
If $x,y$ are
  length-$n$ strings, then $xy\in \EQONES^{=2n}$ if and only if $xy'\in
  \EXMAJEVEN$, where $y'$ is the bitwise complement of $y$. Thus a
  depth $d$ proof system for $\EQONES$ implies a depth $d+1$ proof system
  for $\EXMAJEVEN$.

\item Let $G_1,G_2$ be $n$-node isomorphic graphs with adjacency
  matrices $A_1,A_2$. Then $(A_1,A_2)$ is in $\GI^{=2n^2}$. Let $y_b$
  be the string appearing on the diagonal of $A_b$. Then $y_1y_2 \in
  \EQONES^{=2n}$.

Conversely, for each $xy \in \EQONES^{=2n}$ where $|x|=|y|=n$, the pair
$(\Diag(x),\Diag(y))$ is in $\GI^{=2n^2}$. (For an $n$-bit vector $w$,
$\Diag(w)$ is the $n\times n$ matrix with $w$ on the diagonal and
zeroes elsewhere.)

Thus a depth $d$ proof system for $G$ implies a depth $d$ proof
system for $\EQONES$.
\end{enumerate}
\qed\end{proof}

\section{Discussion}
\label{sec:discussion}

For $\MAJ$, we have given a proof system with $O(\log \log n)$ depth
(and $O(1)$ alternations), and it is known from \cite{BDKMSSTV13} that
$\omega(1)$ depth is needed. Can this gap between the upper and lower
bounds be closed? 

Can we generalize the idea we use in Theorem \ref{thm:ustconn} and
apply it to other languages? In particular, can we obtain good upper bounds using this
technique for the language of $s$-$t$ connected directed graphs? From
the results of \cite{BDKMSSTV13} and this paper, we know languages
complete for $\NC{1}$, $\L$, $\P$ and $\NP$ with $\NC{0}$ proof
systems. A proof system for $\REACH$ would bring  $\NL$ into this
list.

Our construction from Theorem~\ref{thm:regularPS} can be generalized to work for languages accepted
by growing-monoids or growing-non-uniform-automata with
poly-log growth rate (see eg \cite{BLM93}). Can we obtain good upper bounds for
linearly growing automata?

In \cite{KL12}, proof systems computable in DLOGTIME are
investigated. The techniques used there seem quite different from
those that work for small-depth circuits, especially $\pAC{0}$. Though
in both cases each output bit can depend on at most $\poly \log n$
input bits, the circuit can pick an arbitrary set of $\poly \log n$
bits whereas a DLOGTIME proof system needs to write the index of each
bit on the index tape using up $\log n$ time.

\bibliography{biblio}

\begin{thebibliography}{10}

\bibitem{Cook71}
Cook, S.A.:
\newblock The complexity of theorem-proving procedures.
\newblock In: Proceedings of the annual ACM symposium on Theory of Computing.
  (1971)  151--158

\bibitem{Hir10}
Hirsch, E.A.:
\newblock Optimal acceptors and optimal proof systems.
\newblock In: Proceedings of 7th Annual Conference on Theory and Applications
  of Models of Computation TAMC, Springer (2010)

\bibitem{HI10}
Hirsch, E.A., Itsykson, D.:
\newblock On optimal heuristic randomized semidecision procedures, with
  application to proof complexity.
\newblock In: Proceedings of 27th International Symposium on Theoretical
  Aspects of Computer Science, STACS. (2010)  453--464

\bibitem{CK07}
Cook, S.A., Kraj\'{\i}\v{c}ek, J.:
\newblock Consequences of the provability of {NP $\subseteq$ P/poly}.
\newblock Journal of Symbolic Logic \textbf{72}(4) (2007)  1353--1371

\bibitem{Pudlak09}
Pudl{\'a}k, P.:
\newblock Quantum deduction rules.
\newblock Ann. Pure Appl. Logic \textbf{157}(1) (2009)  16--29 See also ECCC
  TR07-032.

\bibitem{CR79}
Cook, S.A., Reckhow, R.A.:
\newblock The relative efficiency of propositional proof systems.
\newblock Journal of Symbolic Logic \textbf{44}(1) (1979)  36--50

\bibitem{BP01}
Beame, P., Pitassi, T.:
\newblock Propositional proof complexity: Past, present, and future.
\newblock In: Current Trends in Theoretical Computer Science.
\newblock World Scientific (2001)  42--70

\bibitem{Seg07}
Segerlind, N.:
\newblock The complexity of propositional proofs.
\newblock Bulletin of Symbolic Logic \textbf{13}(4) (2007)  417--481

\bibitem{FSS84}
Furst, M.L., Saxe, J.B., Sipser, M.:
\newblock Parity, circuits, and the polynomial-time hierarchy.
\newblock Mathematical Systems Theory \textbf{17}(1) (1984)  13--27

\bibitem{Has86}
H{\aa}stad, J.:
\newblock Almost optimal lower bounds for small depth circuits.
\newblock In: Proceedings of the 18th Annual ACM Symposium on Theory of
  Computing STOC. (1986)  6--20

\bibitem{BDKMSSTV13}
Beyersdorff, O., Datta, S., Krebs, A., Mahajan, M., Scharfenberger-Fabian, G.,
  Sreenivasaiah, K., Thomas, M., Vollmer, H.:
\newblock Verifying proofs in constant depth.
\newblock ACM Trans. Comput. Theory \textbf{5}(1) (May 2013)  2:1--2:23 See
  also ECCC TR012-79. A preliminary version appeared in \cite{BDMSSTV11}.

\bibitem{CM01}
Cryan, M., Miltersen, P.B.:
\newblock On pseudorandom generators in {NC$^0$}.
\newblock In: Proceedings of 26th Symposium on Mathematical Foundations of
  Computer Science. (2001)  272--284

\bibitem{Bar89}
Barrington, D.:
\newblock Bounded-width polynomial size branching programs recognize exactly
  those languages in {NC$^1$}.
\newblock Journal of Computer and System Sciences \textbf{38} (1989)  150--164

\bibitem{BT88}
Barrington, D., Th\'erien, D.:
\newblock Finite monoids and the fine structure of {NC}$^1$.
\newblock Journal of the Association of Computing Machinery \textbf{35} (1988)
  941--952

\bibitem{Rei2008}
Reingold, O.:
\newblock Undirected connectivity in log-space.
\newblock J. ACM \textbf{55}(4) (2008) (Originally appeared in STOC '05).

\bibitem{Veb1912}
Veblen, O.:
\newblock An application of modular equations in analysis situs.
\newblock Annals of Mathematics \textbf{14}(1/4) (1912)  pp. 86--94

\bibitem{Jus72}
Justesen, J.:
\newblock Class of constructive asymptotically good algebraic codes.
\newblock Information Theory, IEEE Transactions on \textbf{18}(5) (1972)
  652--656

\bibitem{BLM93}
Bedard, F., Lemieux, F., McKenzie, P.:
\newblock Extensions to {Barrington}'s {M}-program model.
\newblock Theoretical Computer Science \textbf{107} (1993)  31--61

\bibitem{KL12}
Krebs, A., Limaye, N.:
\newblock Dlogtime-proof systems.
\newblock Electronic Colloquium on Computational Complexity (ECCC) \textbf{19}
  (2012)  186

\end{thebibliography}
\newpage
\end{document}